\RequirePackage{amsmath}
\documentclass[11pt]{llncs}
\usepackage[T1]{fontenc}
\usepackage{amssymb}
\usepackage{amsfonts}
\usepackage{dsfont}
\usepackage[colorlinks=true, citecolor=blue]{hyperref}
\usepackage{color,graphicx,setspace}
\usepackage{geometry}
\usepackage{paralist}
\usepackage{todonotes}
\usepackage{comment} 
\usepackage[dvipsnames]{xcolor}
\usepackage{thmtools}
\usepackage{thm-restate}
\usepackage{cleveref}

\usepackage{amsmath}
\usepackage{tikz}
\usetikzlibrary{arrows.meta}
\usetikzlibrary{decorations.pathreplacing}

\newif\ifanonsubmission
\anonsubmissionfalse

\ifanonsubmission
\newcommand{\Dvir}[1]{}
\newcommand{\Tomohiro}[1]{}
\newcommand{\Danny}[1]{}
\newcommand{\Karl}[1]{}
\else
\newcounter{dvircounter}
\newcommand{\Dvir}[1]{\noindent\textcolor{ForestGreen}{$\ll$\textsf{#1}$\gg$\marginpar{\tiny\bf \textcolor{ForestGreen} {\refstepcounter{dvircounter}Dvir~\thedvircounter}}}}
\newcommand{\Tomohiro}[1]{\noindent\textcolor{blue}{$\ll$\textsf{#1}$\gg$\marginpar{\tiny\bf \textcolor{blue}{Tomohiro}}}}
\newcommand{\Danny}[1]{\noindent\textcolor{red}{$\ll$\textsf{#1}$\gg$\marginpar{\tiny\bf \textcolor{red}{Danny}}}}
\newcommand{\Karl}[1]{\noindent\textcolor{Plum}{$\ll$\textsf{#1}$\gg$\marginpar{\tiny\bf \textcolor{Plum}{Karl}}}}
\fi

\newcommand*{\const}{4}

\title{Generic Swapping for Lawler-Moore Speedups}

\ifanonsubmission
\author{Anonymous Submission}
\institute{}
\else
\author{
Karl Bringmann\inst{1} \and Danny Hermelin\inst{2} \and Tomohiro Koana\inst{3} \and Dvir Shabtay \inst{2}}
\institute{ETH Zurich, Switzerland\thanks{Part of this work was done while K.B.~was affiliated to Saarland University and Max Planck Institute for Informatics, Saarbrücken, Germany, where this work was part of the project TIPEA that has received funding from the European Research Council (ERC) under the European Unions Horizon 2020 research and innovation programme (grant agreement No. 850979).} \\
\email{karl.bringmann@inf.ethz.ch}
\and
Department of Industrial Engineering and Management, Ben-Gurion University of the Negev, Israel
\email{hermelin@bgu.ac.il, dvirs@bgu.ac.il}
\and
Graduate School of Information Science and Technology, The University of Tokyo, Japan\thanks{T.K. was supported in part by JST CREST Grant Number JPMJCR24Q2 and JST ERATO Grant Number JPMJER2301}\\ \email{tomohiro.koana@gmail.com}
}
\fi

\begin{document}

\sloppy
\maketitle

\begin{abstract} 
The Lawler-Moore dynamic programming framework is a cornerstone of scheduling theory, providing a robust approach for solving various parallel machine problems. The framework relies on two essential properties: a regular objective function that is monotone in the completion times of jobs, and a fixed priority ordering, such as Smith's Rule or Jackson's Rule, which dictates the processing sequence on each machine. For fundamental objectives, including minimizing total weighted completion time ($Pm||\sum w_jC_j$), maximum lateness ($Pm||L_{\max}$), and the weighted number of tardy jobs ($Pm||\sum w_jU_j$), this framework yields running times of $O(P^{m-1} \cdot n)$ for the first two problems and $O(P^m \cdot n)$ for the third, where~$P$ denotes the total processing time of all jobs. Recent conditional lower bounds under the Strong Exponential Time Hypothesis (SETH) suggest that the dependence on~$P$ is essentially optimal, but they still leave open the possibility of improving the dependence on the maximum processing time~$p_{\max}$ of any single job.

\hspace{12pt} For the problem of minimizing the total processing time of tardy jobs $Pm||\sum p_jU_j$, this possibility of improving the dependence on $p_{\max }$ was recently realized by Klein, Polak, and Rohwedder (SODA 2023), who bypassed the dependence on $P$ entirely. Their approach relies on a proximity argument to show that optimal schedules maintain a load difference of $O(p_{\max}^2)$ between machines for any prefix set of jobs. In this paper, we show that this structural property can be generalized to a broader class of classic parallel machine scheduling problems handled by the Lawler-Moore framework. While their analysis uses a swapping argument designed for the specific case where $w_j = p_j$, we face richer objective functions - most notably minimizing total weighted completion time $Pm||\sum w_jC_j$, whose objective couples processing times and weights and where a single job swap has a cascading effect on all subsequent completion times. To manage these complexities, we develop an extended swapping argument that tracks these global variations in the objective. This allows us to prove that an optimal schedule preserving the $O(p_{\max}^2)$ load-bounding property still exists for these objectives, enabling state-pruning at each step of the dynamic program to replace the runtime dependence on $P$ with $p_{\max}^2$. Consequently, writing $w_{\max}:=\max_j w_j$, we solve $Pm||\sum w_jC_j$ in $O(\min\{p_{\max},w_{\max}\}^{2m-2} \cdot n)$ time, $Pm||L_{\max}$ in $O(p_{\max}^{2m-2} \cdot n)$ time, and $Pm||\sum w_jU_j$ in $O(p_{\max}^{2m-2} \cdot P \cdot n) \le O(p_{\max}^{2m-1} \cdot n^2)$ time. These algorithms strictly improve upon the original Lawler-Moore bounds whenever $p_{\max} = o(\sqrt{P})$. In particular, the weighted-completion-time and maximum-lateness algorithms are near-linear when $p_{\max}$ is polylogarithmic in~$n$, and the weighted-completion-time algorithm is also near-linear when $w_{\max}$ is polylogarithmic in~$n$.

\end{abstract}

\newpage
\pagenumbering{arabic}
\section{Introduction}

The Lawler-Moore algorithm~\cite{LawlerMoore} stands as a cornerstone of scheduling theory, offering a robust dynamic programming framework for solving a wide range of parallel machine problems in pseudo-polynomial time. This framework is particularly effective for problems characterized by \emph{regular} objective functions that are monotone in the completion times of jobs. Additionally, it requires an easily computed \emph{priority ordering} that dictates the processing sequence once jobs are assigned to their respective machines. By leveraging these two properties, the algorithm formulates the problem as a structured set of recursive subproblems, yielding an efficient dynamic program.

To illustrate the Lawler-Moore approach, consider the problem of minimizing the total weighted completion time on a set of $m=O(1)$ identical parallel machines, denoted by $Pm||\sum w_j C_j$ in standard scheduling notation~\cite{Graham1969}. We are given a set of jobs $J=\{1,\ldots,n\}$, where job~$j$ has an integer \emph{processing time} $p_j \in \mathbb{N}$ and an integer \emph{weight} $w_j \in \mathbb{N}$ representing its relative importance or holding cost. The goal is to schedule the jobs on parallel identical machines $M_1,\ldots,M_m$ so as to minimize their total weighted completion time. By Smith’s Rule~\cite{Smith1956}, there exists an optimal schedule in which the jobs on each machine are processed in non-increasing order of their efficiencies $e_j:=w_j/p_j$. Consequently, a schedule $\sigma$ for $J$ is fully defined by a partition $J = J_1(\sigma) \cup \ldots \cup J_m(\sigma)$ (i.e., a family of $m$ subsets of jobs such that each job appears in exactly one subset), where $J_i(\sigma)$ is the set of jobs assigned to machine~$M_i$. The objective is to find a schedule that minimizes $\sum_j w_j C_j$, where $C_j$ is the completion time of job~$j$ in the Smith-ordered sequence on its assigned machine.

In this context, the Lawler-Moore framework leverages Smith’s Rule~\cite{Smith1956}, which says that the jobs on each machine are processed in non-increasing order of efficiency $e_j:=w_j/p_j$ in any optimal schedule. The dynamic program then defines a state $T_j[P_1,\ldots,P_{m-1}]$ as the minimum cost of scheduling the first $j$ jobs, presorted by efficiency, where $P_i$ represents the total processing time (load) assigned to machine~$M_i$. Consider some state $T_j[P_1,\ldots,P_{m-1}]$. Since the jobs are processed according to Smith’s Rule, we can assume without loss of optimality that no job in $\{1,\ldots,j-1\}$ is scheduled after job~$j$. Consequently, if $j$ is assigned to machine $M_i$, it completes at time $C_j=P_i$; if it is assigned to~$M_m$, it completes at time $C_j=P(J_j)-(P_1+\cdots+P_{m-1})$, where $P(J_j):=p_1+\cdots + p_j$ is the total processing time of the first $j$ jobs. This directly leads to the Lawler-Moore recurrence relation:
$$
T_j[P_1,\ldots,P_{m-1}] = \min
\begin{cases}
\quad T_{j-1}[P_1 - p_j,\ldots,P_{m-1}] \,+\, w_j P_1,\\ 
\quad \vdots\\
\quad T_{j-1}[P_1,\ldots,P_{m-1} - p_j] + w_j P_{m-1},\\
\quad T_{j-1}[P_1,\ldots,P_{m-1}] + w_j \cdot (P(J_j)- \sum_{i=1}^{m-1} P_i), \\
\end{cases}
$$
where $T_0[P_1,\ldots,P_{m-1}]=0$ if $P_1,\ldots,P_{m-1}=0$, and $T_0[P_1,\ldots,P_{m-1}]=\infty$ otherwise.

With slight modifications, the Lawler-Moore framework also applies to several other classic scheduling problems, most notably:
\begin{itemize}
\item The $Pm||L_{\max}$ problem, which asks to minimize the maximum lateness. Here each job $j \in J$ has a \emph{due date} $d_j \in \mathbb{N}$, and the \emph{lateness} of job~$j$ is defined as $L_j=C_j-d_j$.
\item The $Pm||\sum w_jU_j$ problem, which asks to minimize the total weighted number of tardy jobs. A job $j \in J$ is tardy if its completion time~$C_j$ exceeds its due date; this is represented by a binary variable $U_j$ that equals 1 iff $C_j > d_j$.
\end{itemize}
For both objectives, Jackson’s Earliest Due Date (EDD) rule~\cite{Jackson1955} provides the required priority ordering on each of the $m$ machines. For the $Pm||L_{\max}$ problem, there always exists an optimal schedule in which the jobs assigned to each machine are sequenced according to the EDD rule. In contrast, for the $Pm||\sum w_jU_j$ problem, there exists an optimal schedule in which the early jobs on each machine are processed first in EDD order, followed by the tardy jobs in an arbitrary order. Equivalently, the tardy jobs may be viewed as discarded.

While the recurrence for $Pm||L_{\max}$ mirrors the $(m-1)$-dimensional state space of the weighted completion time problem, the $Pm||\sum w_jU_j$ problem requires tracking the loads on all $m$ machines, because the load on the $m$-th machine is no longer implicitly determined once tardy jobs may be discarded. This leads to the following classical result, where~$P=\sum_{j\in J}p_j$:
\begin{theorem}[\cite{LawlerMoore}]
\label{thm:LawlerMoore}%
The $Pm||\sum w_jC_j$ and $Pm||L_{\max}$ problems can both be solved in $O(P^{m-1} \cdot n)$ time, while the $Pm||\sum w_jU_j$ problem can be solved in $O(P^m \cdot n)$ time.    
\end{theorem}

Despite its foundational status, the Lawler-Moore framework has remained essentially stagnant for almost sixty years. Since its introduction, the $O(P^{m-1} \cdot n)$ time complexity has served as an unchallenged baseline for pseudo-polynomial algorithms in parallel machine environments. Recent conditional lower bounds reinforce this baseline, showing that, assuming the Strong Exponential Time Hypothesis (SETH), the running time in Theorem~\ref{thm:LawlerMoore} is essentially optimal for all three problems above.
\begin{theorem}[\cite{BringmannDW26}]
\label{thm:LowerBound}%
Assuming SETH, for any $\varepsilon > 0$, the $Pm||\sum w_jC_j$ and $Pm||L_{\max}$ problems cannot be solved in $P^{m-1-\varepsilon} \cdot 2^{o(n)}$ time, and the $Pm||\sum w_jU_j$ problem cannot be solved in $P^{m-\varepsilon} \cdot 2^{o(n)}$ time.    
\end{theorem}

Since the total processing time $P$ can be as large as $p_{\max} \cdot n$, where $p_{\max}$ is the maximum processing time of any single job, the most natural path toward algorithmic improvement is to refine the dependence on $p_{\max}$. While structural approaches like $n$-fold integer programming~\cite{KnopKoutecky18} or thin constraint matrices~\cite{HermelinShabtay2026} offer parameterization tools for related settings, they either lack efficient pseudo-polynomial running times, or scale poorly when instances contain arbitrary weights, non-linear objectives, or multiple distinct due dates. 

A major breakthrough in improving the dependence on $p_{\max}$ was recently achieved by Klein, Polak, and Rohwedder~\cite{Klein0R23} for the problem of minimizing the total processing time of tardy jobs $Pm||\sum p_jU_j$. This setting corresponds to the special case of the $Pm||\sum w_jU_j$ problem where weights are equal to processing times, \emph{i.e.} $w_j = p_j$ for each job $j \in J$. By introducing a proximity argument and a specialized job-swapping technique, they proved a fundamental structural property: There always exists an optimal schedule where the load difference between any two machines is bounded by $O(p_{\max}^2)$ for any prefix set of jobs (see Fig.~\ref{fig:Structure}). This bounded prefix-load property allowed them to prune redundant states at each step of the Lawler-Moore dynamic program, bypassing the traditional dependence on the total processing time $P$ entirely, and yielding the following bound:
\begin{theorem}[\cite{Klein0R23}]
The $Pm||\sum p_jU_j$ problem can be solved in $O(p^{2m-2}_{\max} \cdot n)$ time.
\end{theorem}

\subsection{Our results}

The main contribution of this paper is a generalized framework for accelerating the classical Lawler-Moore algorithm, showing that the $O(p_{\max}^{2m-2} \cdot n)$ runtime bound recently obtained for the $Pm||\sum p_jU_j$ problem~\cite{Klein0R23} can also be achieved for more complex parallel machine objectives. Our framework yields the following results, where $w_{\max}:=\max_{j\in J}w_j$ denotes the maximum weight:
\begin{theorem}
\label{thm:main}%
The $Pm||\sum w_jC_j$ problem can be solved in $O(\min\{p_{\max},w_{\max}\}^{2m-2} \cdot n)$ time, the $Pm||L_{\max}$ problem can be solved in $O(p^{2m-2}_{\max} \cdot n)$ time, and the $Pm||\sum w_jU_j$ problem can be solved in time $O(p^{2m-2}_{\max} \cdot P \cdot n) \le O(p^{2m-1}_{\max} \cdot n^2)$.
\end{theorem}

For total weighted completion time, the stronger dependence on $\min\{p_{\max},w_{\max}\}$ follows by combining our algorithm with a processing-time/weight duality observed by Hall and Chudak and recorded by Goemans and Williamson~\cite[pp.~283--284]{GoemansWilliamson00}. We state and prove the parallel-machine form of this duality in Lemma~\ref{lem:WCT-duality}.

A central component of our framework is an extended swapping argument that generalizes the argument previously used for $Pm||\sum p_jU_j$. Swapping arguments that exchange the positions of jobs in a schedule have long been used to prove the optimality of priority rules such as Smith’s Rule and Jackson’s Rule. However, whereas these traditional arguments typically focus on pairs of jobs, both the approach in~\cite{Klein0R23} and our framework swap subsets of jobs using a proximity argument driven by a lemma from additive combinatorics. Specifically, this lemma states that sufficiently large multisets over a common set of integers $\{1,\ldots,U\}$ must have submultisets whose elements have an equal sum.

\begin{figure}[h]
\centering
\begin{tikzpicture}[x=1.2cm, y=1.2cm]
\tikzset{job/.style={draw=black, fill=gray!15, thin}}
\tikzset{futurejob/.style={draw=black, fill=gray!60, thin}} 
\draw[thick, ->] (0,0) -- (8,0) node[right] {Time};
\draw[thick, ->] (0,0) -- (0,4.5) node[above] {Machines};
\node at (-0.5, 4) {$M_4$};
\foreach \x/\w in {0/0.4, 0.4/0.3, 0.7/0.6, 1.3/0.2, 1.5/0.5, 2.0/0.4}
\draw[job] (\x, 3.7) rectangle (\x+\w, 4.3);
\foreach \x/\w in {2.4/0.5, 2.9/0.3, 3.2/0.6, 3.8/0.4, 4.2/0.7, 4.9/0.3, 5.2/0.5, 5.7/0.4}
\draw[futurejob] (\x, 3.7) rectangle (\x+\w, 4.3);
\node at (-0.5, 3) {$M_3$};
\foreach \x/\w in {0/0.5, 0.5/0.7, 1.2/0.4, 1.6/0.6, 2.2/0.3, 2.5/0.5, 3.0/0.3}
\draw[job] (\x, 2.7) rectangle (\x+\w, 3.3);
\foreach \x/\w in {3.3/0.4, 3.7/0.6, 4.3/0.3, 4.6/0.8, 5.4/0.4, 5.8/0.5, 6.3/0.3}
\draw[futurejob] (\x, 2.7) rectangle (\x+\w, 3.3);
\node at (-0.5, 2) {$M_2$};
\foreach \x/\w in {0/0.3, 0.3/0.5, 0.8/0.2, 1.0/0.4, 1.4/0.6, 2.0/0.3, 2.3/0.5, 2.8/0.4}
\draw[job] (\x, 1.7) rectangle (\x+\w, 2.3);
\foreach \x/\w in {3.2/0.6, 3.8/0.4, 4.2/0.5, 4.7/0.3, 5.0/0.7, 5.7/0.6, 6.3/0.4}
\draw[futurejob] (\x, 1.7) rectangle (\x+\w, 2.3);
\node at (-0.5, 1) {$M_1$};
\foreach \x/\w in {0/0.4, 0.4/0.2, 0.6/0.5, 1.1/0.3, 1.4/0.4, 1.8/0.6, 2.4/0.2, 2.6/0.5, 3.1/0.4, 3.5/0.6, 4.1/0.3, 4.4/0.8}
\draw[job] (\x, 0.7) rectangle (\x+\w, 1.3);    
\foreach \x/\w in {5.2/0.4, 5.6/0.5, 6.1/0.6, 6.7/0.4, 7.1/0.3}
\draw[futurejob] (\x, 0.7) rectangle (\x+\w, 1.3);
\draw[dashed, gray] (2.4, 4.3) -- (2.4, -0.5); 
\draw[dashed, gray] (5.2, 1.3) -- (5.2, -0.5);     
\draw[thick, <->] (2.4, -0.3) -- (5.2, -0.3) 
node[midway, below] {$O(p_{\max}^2)$};
\end{tikzpicture}
\caption{Gantt chart depicting machine loads. Light gray blocks represent the jobs $\{1, \dots, j\}$ already considered in the dynamic programming state, while dark gray blocks represent future jobs in the priority sequence. The state space is pruned by ensuring the load difference between any two machines at step $j$ remains within $O(p_{\max}^2)$.}
\label{fig:Structure}
\end{figure}
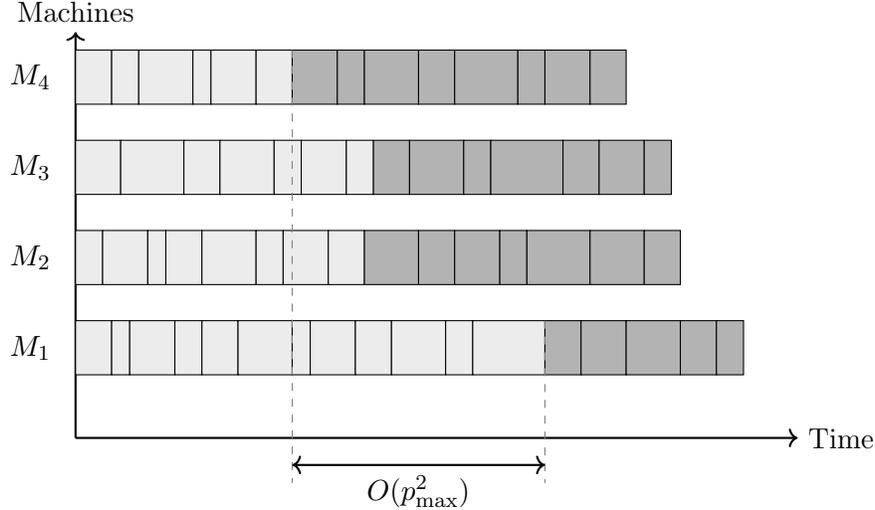

While the approach in~\cite{Klein0R23} utilized this additive combinatorics lemma via a swapping technique designed strictly for minimizing tardy processing times, our framework extends the mechanism to accommodate these more complex objective functions. This allows us to prove that an optimal schedule preserving the $O(p_{\max}^2)$ load-bounding property still exists even under the effects of arbitrary weights and priority-driven sequencing constraints (see Fig.~\ref{fig:Structure}). By identifying these equal-sum subsets under more complex objectives, we show that any partial schedule that strays too far from a balanced configuration can be transformed into one that is just as good, or even better, safely pruning those states in the Lawler-Moore dynamic program.

\subsection{Related work}

We briefly review the state-of-the-art results for the $Pm||\sum w_jC_j$, $Pm||L_{\max}$, and $Pm||\sum w_jU_j$ problems, as well as their generalizations to arbitrary $m$. Our discussion focuses exclusively on exact algorithms and hardness results for these three problems.

While the single-machine $1||\sum w_jC_j$ problem is efficiently solved by Smith’s Rule in $O(n \log n)$ time, this result only extends to the parallel machine unweighted counterpart for an arbitrary number of machines (the $P||\sum C_j$ problem)~\cite{ConwayMM1967}. The weighted version $Pm||\sum w_jC_j$ is NP-hard for $m \geq2$, even when $p_j =w_j$ for all jobs $j \in J$~\cite{BrunoCoffman}, while $P||\sum w_jC_j$ (where~$m$ is arbitrary) is strongly NP-hard~\cite{GareyJohnson}. On the other hand, the problem is only weakly NP-hard for fixed~$m$, due to Lawler and Moore’s pseudo-polynomial algorithm~\cite{LawlerMoore}. Furthermore, Knop and Kouteck{\'{y}}~\cite{KnopKoutecky18} showed that $P||\sum w_j C_j$ is fixed-parameter tractable (FPT) for parameter $k= \max \{p_{\max},w_{\max}\}$, by presenting an algorithm for the problem running in $k^{O(k^2)} \cdot n^{O(1)}$ time. 

The single-machine $1||L_{\max}$ problem is optimally solved by Jackson’s rule in $O(n \log n)$ time~\cite{Jackson1955}. However, the problem is significantly more difficult on identical parallel machines, as $P||L_{\max}$ generalizes the classical makespan minimization problem $P||C_{\max}$, even when jobs have precedence constraints and non-trivial release times~\cite{LagewegEtAl1976}. Since $P||C_{\max}$ is NP-hard for any fixed $m \geq 2$~\cite{LenstraEtAl77}, and strongly NP-hard for arbitrary~$m$, this also holds for $P||L_{\max}$. On the positive side, the problem remains solvable in $O(P^{m-1} \cdot n)$  time via the Lawler-Moore framework~\cite{LawlerMoore}. More recently, Hermelin and Shabtay~\cite{HermelinShabtay2026} utilized thin ILPs to provide an algorithm running in $\widetilde{O}(p^{d_{\#}(m+1)}_{\max}+d_{\#}n)$ time, where $d_{\#}$ denotes the number of distinct due dates, and $\widetilde{O}$ is used to suppress polylogarithmic factors.

Regarding the weighted number of tardy jobs, the unweighted single machine $1||\sum U_j$ problem is solved in $O(n \log n)$ time by the Moore-Hodgson algorithm~\cite{Moore1968}. On parallel machines the problem is more difficult, as $P||\sum w_jU_j$ generalizes $P||C_{\max}$, making it weakly NP-hard whenever $m \geq 2$~\cite{LenstraEtAl77}, and strongly NP-hard for arbitrary~$m$~\cite{GJ78}. The weighted $1||\sum w_jU_j$ problem is also NP-hard as it generalizes $P2||C_{\max}$. As mentioned above, $Pm|| \sum w_jU_j$ is solvable in $O(P^m n)$ time using the Lawler-Moore framework~\cite{LawlerMoore}. This running time can be improved for the single machine case and small values of $d_{\#}$ using $(\max,+)$-convolutions~\cite{HermelinMS24}. For the special case of minimizing the total processing time of tardy jobs (the $Pm||\sum p_jU_j$ problem where $w_j = p_j$), a series of recent improvements~\cite{BringmannFHSW22,SchieberS23} culminated in a proximity-driven state-pruning technique by Klein, Polak, and Rohwedder~\cite{Klein0R23}, yielding a runtime of $O(p_{\max}^{2m-2} \cdot n)$. Fischer and Wennmann~\cite{FischerWennmann25} later presented an alternative algorithm for this restricted problem variant running in $O(P^m)$ time. Finally, $1||\sum w_jU_j$ is FPT for any of the parameter pairs $d_{\#}+p_{\#}$, $d_{\#}+w_{\#}$, or $p_{\#}+w_{\#}$ (where~$p_{\#}$ and~$w_{\#}$ respectively denote the number of distinct processing times and weights in the input)~\cite{HermelinKPShabtay21}, a result extending also to a parameterized number of parallel machines, but remains $W[1]$-hard or NP-hard for these as single parameters~\cite{HeegerH24}.

\section{Preliminaries}

We consider a set of jobs $J=\{1,\ldots,n\}$, given offline, to be scheduled non-preemptively on a set of~$m$ identical parallel machines $M_1,\ldots,M_m$, for some fixed constant $m=O(1)$. The \emph{processing time}, \emph{weight}, and \emph{due date} of job $j \in J$ are denoted by $p_j$, $w_j$, and~$d_j$, respectively; we assume that $p_j$ and $w_j$ are positive integers. The \emph{efficiency} of job~$j$ is defined by~$e_j:=w_j/p_j$. We denote the total processing time of all jobs by $P:= \sum_j p_j$, the maximum processing time by $p_{\max} = \max_j p_j$, and the maximum weight by $w_{\max}=\max_j w_j$. For a subset of jobs $J' \subseteq J$, we write $P(J')=\sum_{j\in J'}p_j$ for the total processing time of the jobs in $J'$ (so, in particular, $P=P(J)$).

Throughout this paper, we assume that the jobs in $J$ are preordered according to a fixed priority ordering consistent with the specific objective function being minimized.
\begin{itemize}
\item For $Pm||\sum w_jC_j$, jobs are ordered such that $e_1 \geq e_2 \geq \cdots \geq e_n$.
\item For $Pm||\sum w_jU_j$ and $Pm||L_{\max}$, jobs are ordered $d_1 \leq d_2 \leq \cdots \leq d_n$. 
\end{itemize}
Under this assumption, whenever two jobs $j_1 < j_2$ are assigned to the same machine in an optimal schedule, there exists an optimal schedule in which $j_1$ is processed before $j_2$. We refer to such schedules as \emph{proper schedules}. 

Since the processing order is predetermined in a proper schedule, any such schedule $\sigma$ is determined by a partition $\sigma=\{J_1(\sigma),\ldots,J_m(\sigma)\}$, where $J_i(\sigma) \subseteq J$ is the set of jobs assigned to machine $M_i$ for all $i \in \{1,\ldots,m\}$. Here, we say that $J_1(\sigma),\ldots,J_m(\sigma) \subseteq J$ is a partition of $J$ if each job in $J$ appears in exactly one subset $J_i(\sigma)$. For $j \in J$, we let $J_j=\{1,\ldots,j\}$ denote the set of all jobs that precede $j$ in the priority ordering, and define $J_{i,j}(\sigma):=J_i(\sigma) \cap J_j$. The \emph{load on~$M_i$ at step~$j$ of~$\sigma$} is the total processing time of jobs in $J_{i,j}(\sigma)$, namely $P(J_{i,j}(\sigma))$. Because the total load on all machines at step~$j$ equals $P(J_j)$, the load of the $m$-th machine is implicitly determined by the loads of the first $m-1$ machines: 
$$
P(J_{m,j}(\sigma))=P(J_j)-\sum^{m-1}_{i=1} P(J_{i,j}(\sigma)).
$$
If job $j$ is assigned to machine $M_i$ in $\sigma$, then the \emph{completion time} $C_j(\sigma)$ of~$j$ is the load on~$M_i$ at step~$j$, \emph{i.e.} $C_j(\sigma)=P(J_{i,j}(\sigma))$. Throughout the paper we consider only \emph{regular} objective functions, \emph{i.e.} objective functions that are non-decreasing in the completion times $C_1(\sigma),\ldots,C_n(\sigma)$. Observe that $\sum w_jC_j$, $L_{\max}$, and $\sum w_jU_j$ are all regular.

\section{Swapping via Additive Combinatorics}
\label{sec:Swap}%

We now describe a general swapping argument that can be applied to various scheduling problems on parallel machines. Let $f$ be an arbitrary regular objective function that admits a priority ordering. Broadly speaking, given any proper schedule $\sigma$ for~$J$, the swap defined below, when admissible, produces an alternative proper schedule~$\sigma''$ for $J$ (the notation $\sigma'$ is reserved for an intermediate schedule produced by the swap). We say that the swap is \emph{optimal} for $f$ if $f(\sigma'') \leq f(\sigma)$ for every proper schedule $\sigma$. The goal of this swapping argument is the following generic structure theorem, which allows us to prune the state space of the Lawler-Moore dynamic program. 
\begin{restatable}{theorem}{genericstructure}
\label{thm:GenericStructure}
Let $f$ be a regular objective function that admits a priority ordering. Suppose that the swap described below is optimal for $f$. Then any $Pm||f$ instance admits an optimal proper schedule~$\sigma$ for~$f$ with 
$$
|P(J_{h,j}(\sigma)) - P(J_{i,j}(\sigma))| \;\leq\; \const p_{\max}^2
$$
for all $1 \leq h \leq i \leq m$ and~$1 \leq j \leq n$.      
\end{restatable}

\subsection{Additive combinatorics lemma}

Our swapping argument relies on the following lemma regarding the existence of submultisets with equal sums. A similar argument was used by Polak, Rohwedder, and W\k{e}grzycki~\cite[Lemma 2.1]{PolakEtAl21} to obtain a fast algorithm for the Knapsack problem. For a multiset of integers $A$, we write $\Sigma(A)$ for the sum of its elements, counted with multiplicity.
\begin{lemma}
\label{lem:equal-sum-set}%
Let $U\in\mathbb{N}$, and let $A,B$ be multisets whose elements lie in $\{1,2,\dots,U\}$. If $|A|,|B| \geq 2U$ then there exist non-empty $A'\subseteq A$ and $B'\subseteq B$ such that $\Sigma(A') = \Sigma(B')$.
\end{lemma}

\begin{proof}
Let $\{a_1,..,a_{2U}\} \subseteq A$ and $\{b_1,..,b_{2U}\} \subseteq B$ be arbitrary $2U$ elements in $A$ and $B$. We construct sets $A_i \subseteq A$ and $B_i \subseteq B$ for each $i \in\{0,1,..,2U\}$ as follows:
\begin{itemize}
\item For $i=0$, we set $A_0 := B_0 := \emptyset$.
\item For $i > 0$, if $\Sigma(A_{i-1}) \leq \Sigma(B_{i-1})$ then we set $A_i := A_{i-1} \cup \{a_i\}$ and $B_i := B_{i-1}$, otherwise we set $A_i :=A_{i-1}$ and $B_i := B_{i-1} \cup \{b_i\}$.
\end{itemize}  
Observe that $-U < \Sigma(A_i) - \Sigma(B_i) \leq U$ for all $i \in\{ 0,1,..,2U\}$. By the pigeonhole principle, since there are $2U$ integers $x$ with $-U < x \leq U$ and $2U+1$ many indices~$i$, there exist $i < j$ with $\Sigma(A_i) - \Sigma(B_i) = \Sigma(A_j) - \Sigma(B_j)$. Rearranging yields $\Sigma(A_j) - \Sigma(A_i) = \Sigma(B_j) - \Sigma(B_i)$, and so $\Sigma(A_j \setminus A_i) = \Sigma(B_j \setminus B_i)$. Set $A' := A_j \setminus A_i$ and $B' := B_j \setminus B_i$. Then $\Sigma(A')=\Sigma(B')$, and since $j > i$ at least one of $A'$ and $B'$ is nonempty. As $\Sigma(A')=\Sigma(B')$, actually both $A',B' \neq \emptyset$. \qed
\end{proof}

\subsection{The swap}

Let $\sigma$ be a proper schedule for $J$. Our swap is \emph{admissible for $\sigma$ at step $j \in J$} if there exists a pair of machine indices $h,i \in \{1,\ldots,m\}$ such that the following properties hold:
\begin{itemize}
\item $j \in J_{h,j}(\sigma)$,
\item $|J_i(\sigma) \setminus J_{i,j}(\sigma)| \;\geq\; 2p_{\max}$, and 
\item $\Delta_{h,i,j}(\sigma) \;:=\; P(J_{h,j}(\sigma)) - P(J_{i,j}(\sigma)) \;\geq \; \const p^2_{\max}$.
\end{itemize}
Suppose these three conditions hold for $\sigma$. Let $t_i := P(J_{i,j}(\sigma))$ denote the time at which the last job from $J_j$ completes on~$M_i$. Since $\sigma$ is proper, every job processed on~$M_i$ after time $t_i$ belongs to $J \setminus J_j=\{j+1,\ldots,n\}$. On the other hand, machine~$M_h$ continues processing jobs from $J_j$ until time $t_h:=t_i+\Delta_{h,i,j}(\sigma)$. Let $J_I$ denote the first $2p_{\max}$ jobs in $J_i(\sigma) \setminus J_{i,j}(\sigma)$, and let $J_H$ denote the last $2p_{\max}$ jobs in $J_{h,j}(\sigma)$. 
Since $P(J_I), P(J_H) \le 2 p_{\max}^2$ and $\Delta_{h,i,j}(\sigma) \ge 4 p_{\max}^2$, all jobs in $J_I$ are completed before any job in $J_H$ is started (see Fig.~\ref{fig:Swap}). 

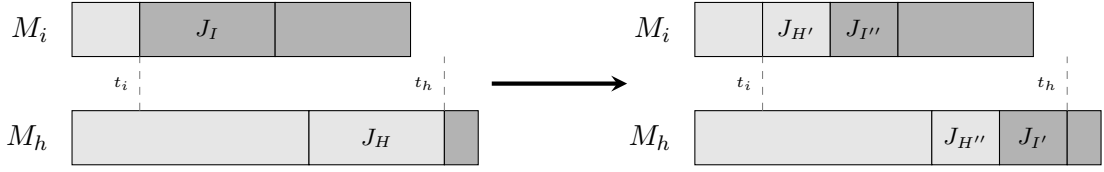
\begin{figure}[h]
\centering
\begin{tikzpicture}[x=0.9cm, y=1.2cm]
\begin{scope}[shift={(0,0)}]
\draw[fill=black!10] (0, 2.2) rectangle (1, 2.8); 
\draw[fill=black!30] (1, 2.2) rectangle (3, 2.8); 
\node[font=\small] at (2, 2.5) {$J_I$};
\draw[fill=black!30] (3, 2.2) rectangle (5, 2.8); 

\draw[fill=black!10] (0, 1) rectangle (3.5, 1.6); 
\draw[fill=black!10] (3.5, 1) rectangle (5.5, 1.6); 
\node[font=\small] at (4.5, 1.3) {$J_H$};
\draw[fill=black!30] (5.5, 1) rectangle (6, 1.6); 

\node[left] at (-0.2, 2.5) {$M_i$};
\node[left] at (-0.2, 1.3) {$M_h$};

\draw[dashed, gray] (1, 1.6) -- (1, 2.2);
\node[left, font=\tiny] at (1, 1.9) {$t_i$};
\draw[dashed, gray] (5.5, 1.6) -- (5.5, 2.2);
\node[left, font=\tiny] at (5.5, 1.9) {$t_h$};
\end{scope}

\draw[->, ultra thick, >=stealth] (6.2, 1.9) -- (8.2, 1.9); 

\begin{scope}[shift={(9.2,0)}]
\draw[fill=black!10] (0, 2.2) rectangle (1, 2.8); 
\draw[fill=black!10] (1, 2.2) rectangle (2, 2.8); 
\node[font=\small] at (1.5, 2.5) {$J_{H'}$};
\draw[fill=black!30] (2, 2.2) rectangle (3, 2.8); 
\node[font=\small] at (2.5, 2.5) {$J_{I''}$};
\draw[fill=black!30] (3, 2.2) rectangle (5, 2.8); 

\draw[fill=black!10] (0, 1) rectangle (3.5, 1.6); 
\draw[fill=black!10] (3.5, 1) rectangle (4.5, 1.6); 
\node[font=\small] at (4, 1.3) {$J_{H''}$};
\draw[fill=black!30] (4.5, 1) rectangle (5.5, 1.6); 
\node[font=\small] at (5, 1.3) {$J_{I'}$};
\draw[fill=black!30] (5.5, 1) rectangle (6, 1.6); 

\node[left] at (-0.2, 2.5) {$M_i$};
\node[left] at (-0.2, 1.3) {$M_h$};

\draw[dashed, gray] (1, 1.6) -- (1, 2.2);
\node[left, font=\tiny] at (1, 1.9) {$t_i$};
\draw[dashed, gray] (5.5, 1.6) -- (5.5, 2.2);
\node[left, font=\tiny] at (5.5, 1.9) {$t_h$};
\end{scope}
\end{tikzpicture}
\caption{A depiction of the intermediate schedule $\sigma'$ obtained from a swap performed between machines~$M_h$ and $M_i$. The lighter gray depicts jobs in $J_j=\{1,\ldots,j\}$. The darker gray depicts jobs in $J \setminus J_j$. In the final schedule $\sigma''$, the jobs of $J_{H'}$ and $J_{I''}$ can move to the left and right on $M_i$, while the jobs of $J_{I'}$ can move to the right on $M_h$.}
\label{fig:Swap}
\end{figure}

Let $H$ and $I$ denote the multisets of processing times of the jobs in $J_H$ and $J_I$, respectively. Note that the elements of $H$ and $I$ lie in $\{1,\ldots,p_{\max}\}$, and that $|H|,|I| \geq 2p_{\max}$. Therefore, by \Cref{lem:equal-sum-set} (with $A=H$, $B=I$, and $U=p_{\max}$), there exist non-empty submultisets $H'\subseteq H$ and $I'\subseteq I$ with $\Sigma(H')=\Sigma(I')$. Let~$J_{H'}$ and~$J_{I'}$ denote the corresponding sets of jobs on $M_h$ and $M_i$, respectively.

We construct an intermediate schedule $\sigma'$ by modifying only $M_h$ and $M_i$. On machine~$M_h$, we replace $J_H$ by first the jobs of $J_{H''}:= J_H \setminus J_{H'}$ in their order in $\sigma$, and then the jobs of $J_{I'}$ in their order in $\sigma$. On machine $M_i$, we replace~$J_I$ by first the jobs of $J_{H'}$ in their order in $\sigma$, and then the jobs of $J_{I''}:=J_I\setminus J_{I'}$ in their order in $\sigma$ (see Fig.~\ref{fig:Swap}). Finally, we sort the jobs on both machines by priority to obtain the proper schedule~$\sigma''$.  

\subsection{Proof of the structure theorem}

We next prove Theorem~\ref{thm:GenericStructure}. Recall that the swap is optimal for~$f$ if $f(\sigma'') \leq f(\sigma)$ for every proper schedule $\sigma$ and every schedule $\sigma''$ obtained from~$\sigma$ by an admissible swap. We call a proper schedule~$\sigma$ \emph{unbalanced at step $j$}, for $j \in J$, if $\Delta_{h,i,j}(\sigma) > \const p^2_{\max}$ for some pair of machines $M_h$ and $M_i$; otherwise, $\sigma$ is \emph{balanced at step~$j$}. To prove Theorem~\ref{thm:GenericStructure}, it suffices to show that, assuming the swap is optimal, there exists an optimal proper schedule for~$J$ that is balanced at every step $j \in J$. The final step $j=n$ is easier and, regardless of whether the swap is optimal, admits a much stronger bound. We call a proper schedule $\sigma$ \emph{leveled} if $\Delta_{h,i,n}(\sigma) \leq p_{\max}$ for all $h,i \in \{1,\ldots,m\}$.

\begin{lemma}
\label{lem:LeveledSchedules}
There exists a leveled optimal proper schedule for $J$. 
\end{lemma}
\begin{proof}
Let $\sigma$ be an optimal proper schedule that is not leveled, and let $M_h$ and $M_i$ be machines of maximum and minimum load in $\sigma$, respectively. Since $\sigma$ is not leveled, we have
$$
\Delta_{h,i,n}(\sigma) \;=\; P(J_h(\sigma)) - P(J_i(\sigma)) \;>\; p_{\max}.
$$ 
Let $j$ be the last job processed on $M_h$. Then $C_j(\sigma) = P(J_h(\sigma))$, and hence $P(J_i(\sigma)) \leq C_j(\sigma) - p_j$. Thus, we can move~$j$ from $M_h$ to the end of~$M_i$ without increasing its completion time. Since $f$ is regular, the resulting schedule is still optimal. Reordering the jobs on~$M_i$ according to the priority ordering restores properness without affecting optimality. Repeating this transformation yields a leveled optimal proper schedule. \qed 
\end{proof}

Next, for a contradiction, assume that no optimal proper schedule for~$J$ is balanced at every step $j \in J$. By Lemma~\ref{lem:LeveledSchedules}, there exists a leveled optimal proper schedule, so let $\Sigma \neq \emptyset$ denote the set of all such schedules. For each schedule $\sigma \in \Sigma$, let $j_{\min}(\sigma)$ be the minimum step at which $\sigma$ is unbalanced. Define
\[
j^* := \max_{\sigma \in \Sigma} j_{\min}(\sigma).
\]
Such a $j^*$ exists by our assumption that no schedule in $\Sigma$ is balanced at every step. Let $\Sigma^* \subseteq \Sigma$ denote the set of all leveled optimal proper schedules $\sigma$ with $j_{\min}(\sigma) = j^*$. For any schedule $\sigma \in \Sigma^*$, define $\delta(\sigma)=\max_{h,i \in \{1,\ldots,m\}} \Delta_{h,i,j^*}(\sigma)$, and let $\delta_{\min} = \min_{\sigma \in \Sigma^*} \delta(\sigma)$. Finally, choose a schedule $\sigma \in \Sigma^*$ with $\delta(\sigma)=\delta_{\min}$. 
We derive a contradiction by constructing a leveled optimal proper schedule $\sigma''$ such that either $j_{\min}(\sigma'') > j^*$, or $j_{\min}(\sigma'') = j^*$ and $\delta(\sigma'') < \delta_{\min}$.

Let $M_h$ and $M_i$ be a pair of machines maximizing the load gap of $\sigma$ at step $j^*$; that is, a pair of machines with $\Delta_{h,i,j^*}(\sigma) = \delta(\sigma) \geq \const p^2_{\max}$. Since $\Delta_{h',i',j}(\sigma) <  \const p^2_{\max}$ for all $h',i' \in \{1,\ldots,m\}$ and all $j < j^*$, machine~$M_h$ is the unique machine of maximum load at step~$j^*$, and $j^*$ is the last job from $J_{j^*}$ processed on $M_h$.
Moreover, we claim that $|J_i(\sigma) \setminus J_{i,j^*}(\sigma)| \geq 2p_{\max}$. Otherwise, $P(J_i(\sigma) \setminus J_{i,j^*}(\sigma)) < 2p^2_{\max}$, and~$\sigma$ would not be leveled, since
\begin{align*}
P(J_h(\sigma)) - P(J_i(\sigma)) &\;\geq\; P(J_{h,j^*}(\sigma)) - P(J_{i,j^*}(\sigma)) - P(J_i(\sigma) \setminus J_{i,j^*}(\sigma))\\   
&\;> \; \const p^2_{\max} - 2p^2_{\max} \;>\; p_{\max}.
\end{align*}
Thus, the swap is admissible at step~$j^*$ of $\sigma$. Let $\sigma''$ denote the proper schedule resulting from the swap. Then $\sigma''$ is optimal by the assumption that the swap is optimal, and moreover the completion time of~$j^*$ decreases in $\sigma''$. 
We next argue that $\sigma''$ is balanced at every step $j < j^*$, which will then yield the desired contradiction.

\begin{lemma}
\label{lem:BalancedAfterSwap}
$\sigma''$ is balanced at every step $j < j^*$.   
\end{lemma}

\begin{proof}
Consider some $j < j^*$, and let $j_0 < j^*$ be the first job, in the priority ordering, that belongs to~$J_{H'}$. Observe that in~$\sigma''$, the position of every job with index smaller than $j_0$ is the same as in $\sigma$, and hence $\sigma''$ is balanced at every step $j < j_0$ by the choice of~$\sigma$. It therefore remains to show that for every $j \in \{j_0,\ldots,j^*-1\}$ and every pair $a,b \in \{1,\ldots,m\}$, we have $\Delta_{a,b,j}(\sigma'') < \const p^2_{\max}$. So fix some $j \in \{j_0,\ldots,j^*-1\}$ and consider all possibilities for $a$ and $b$. Since only the jobs on machines $M_h$ and $M_i$ change position, it suffices to consider the cases where $a \neq b$ and $\{a,b\} \cap \{h,i\} \neq \emptyset$.

First suppose that $a=h$ and $b=i$. Among jobs in $J_j$, some jobs are moved from machine $M_h$ to machine $M_i$, but no jobs are moved from machine $M_i$ to machine $M_h$. Thus, $P(J_{h,j}(\sigma'')) \le P(J_{h,j}(\sigma))$ and $P(J_{i,j}(\sigma'')) \ge P(J_{i,j}(\sigma))$. Consequently, 
\begin{align*}
\Delta_{a,b,j}(\sigma'') = P(J_{h,j}(\sigma'')) - P(J_{i,j}(\sigma''))
\le  P(J_{h,j}(\sigma))- P(J_{i,j}(\sigma))
= \Delta_{h,i,j}(\sigma) \;<\; \const p^2_{\max}.
\end{align*}

Next suppose that $a=i$ and $b=h$. 
Since the jobs in $J_j$ that are moved from machine $M_h$ to machine $M_i$ are a subset of $J_H$, and since $j \le j^*$, we have $P(J_{i,j}(\sigma'')) \le P(J_{i,j}(\sigma)) + P(J_H) \le P(J_{i,j^*}(\sigma)) + P(J_H)$. 
Since $j \ge j_0 \in J_H$, jobs in $J_{h,j^*}(\sigma) \setminus J_H$ stay on machine~$M_h$ and thus $P(J_{h,j}(\sigma'')) \ge P(J_{h,j^*}(\sigma)) - P(J_H)$. Combining this with $P(J_H) \le 2 p_{\max}^2$ and $P(J_{h,j^*}(\sigma)) - P(J_{i,j^*}(\sigma)) = \Delta_{h,i,j^*}(\sigma) \ge 4 p_{\max}^2$ yields
$\Delta_{i,h,j}(\sigma'') = P(J_{i,j}(\sigma'')) - P(J_{h,j}(\sigma'')) \le 0 < 4 p_{\max}^2$.


Now suppose that $a=h$ and $b \in \{1\ldots,m\} \setminus \{h,i\}$. Note that the completion time of any job in~$J_j$ on $M_h$ does not increase in $\sigma''$ compared to $\sigma$, while the position of any job on $M_b$ remains the same as in~$\sigma$. Thus, $P(J_{h,j}(\sigma'')) \leq P(J_{h,j}(\sigma))$ and $P(J_{b,j}(\sigma))=P(J_{b,j}(\sigma''))$, and therefore
\begin{align*}
\Delta_{a,b,j}(\sigma'') \;&=\; P(J_{h,j}(\sigma'')) - P(J_{b,j}(\sigma''))\\ 
\;&\leq\;  P(J_{h,j}(\sigma))- P(J_{b,j}(\sigma))\\ 
\;&=\; \Delta_{h,b,j}(\sigma) \;<\; \const p^2_{\max}.
\end{align*}

Suppose next that $a=i$ and $b \in \{1\ldots,m\} \setminus \{h,i\}$. Since $j < j^*$, every job in~$J_I$ has index larger than~$j$, and therefore the swap can only increase the load of~$M_i$ at step~$j$ by jobs from~$J_{H'}$. Thus,
\[
P(J_{i,j}(\sigma'')) \;\leq\; P(J_{i,j^*}(\sigma)) + P(J_{H'}) \;=\; P(J_{i,j^*}(\sigma)) + P(J_{I'}) \;\leq\; P(J_{i,j^*}(\sigma)) + P(J_I).
\]
On the other hand, by the definition of~$J_H$, every job of~$J_H$, and in particular $j_0$, starts on~$M_h$ in~$\sigma$ no earlier than time~$P(J_{i,j^*}(\sigma))+P(J_I)$. Since $j \geq j_0$, this implies that
\(
P(J_{i,j}(\sigma'')) \;<\; P(J_{h,j}(\sigma)).
\)
Since machine~$M_b$ is unchanged, we have $P(J_{b,j}(\sigma'')) = P(J_{b,j}(\sigma))$. Therefore,
\begin{align*}
\Delta_{a,b,j}(\sigma'') \;=\; P(J_{i,j}(\sigma'')) - P(J_{b,j}(\sigma'')) 
\;<\; P(J_{h,j}(\sigma)) - P(J_{b,j}(\sigma)) 
\;=\; \Delta_{h,b,j}(\sigma) \;<\; \const p^2_{\max},
\end{align*}

Suppose now that $a\in \{1\ldots,m\} \setminus \{h,i\}$ and $b=h$. Since $j \geq j_0$, the set $J_{H'} \cap J_j$ contains~$j_0$. By the definition of~$J_H$, every job of~$J_H$, and in particular~$j_0$, starts on~$M_h$ in~$\sigma$ no earlier than time~$P(J_{i,j^*}(\sigma))+P(J_I)$. Therefore, after removing the jobs of~$J_{H'} \cap J_j$ from~$M_h$, the load on~$M_h$ at step~$j$ in~$\sigma''$ is still at least~$P(J_{i,j^*}(\sigma))+P(J_I)$. Since $P(J_I)>0$ and $j<j^*$, it follows that
\[
P(J_{h,j}(\sigma'')) \;\geq\; P(J_{i,j^*}(\sigma))+P(J_I) \;>\; P(J_{i,j^*}(\sigma)) \;\geq\; P(J_{i,j}(\sigma)).
\]
As machine~$M_a$ is unchanged, we have $P(J_{a,j}(\sigma'')) = P(J_{a,j}(\sigma))$. Hence
\begin{align*}
\Delta_{a,b,j}(\sigma'') \;=\; P(J_{a,j}(\sigma'')) - P(J_{h,j}(\sigma'')) 
\;<\; P(J_{a,j}(\sigma)) - P(J_{i,j}(\sigma)) 
\;=\; \Delta_{a,i,j}(\sigma) \;<\; \const p^2_{\max},
\end{align*}

Finally suppose that $a\in \{1\ldots,m\} \setminus \{h,i\}$ and $b=i$. Then as the swap only adds jobs from $J_{j^*}$ to machine~$M_i$, and none of these jobs is removed from this machine, we have $P(J_{i,j}(\sigma'')) > P(J_{i,j}(\sigma))$. Accordingly, 
$$
\Delta_{a,i,j}(\sigma'') \;<\; \Delta_{a,i,j}(\sigma) \;<\; \const p^2_{\max}.
$$
The lemma thus follows. \qed
\end{proof}

According to Lemma~\ref{lem:BalancedAfterSwap}, the new schedule~$\sigma''$ can only be unbalanced at a step $j \geq j^*$, and so $j_{\min}(\sigma'') \geq j^*$. As $j^*$ is the last job scheduled on $M_h$ in $\sigma$, and its completion time decreases in $\sigma''$, we have $\Delta_{h,i,j^*}(\sigma'') < \Delta_{h,i,j^*}(\sigma)$. Moreover, since $M_h$ is the unique machine of maximum load in $\sigma$ at step~$j^*$, we also have $\Delta_{h,i,j^*}(\sigma'') < \Delta_{h,a,j^*}(\sigma)$ for all $1 \leq a \leq m$. As $\Delta_{h,i,j^*}=\delta(\sigma)$, this implies that either $j_{\min}(\sigma'') > j^*$, or $j_{\min}(\sigma'') = j^*$ but $\delta(\sigma'') < \delta(\sigma) = \delta_{\min}$. The former case is a contradiction to our choice of~$j^*$, while the latter is a contradiction to our choice of~$\delta_{\min}$. This completes the proof of Theorem~\ref{thm:GenericStructure}.

\subsection{A useful corollary}

For our algorithms, it will be useful to derive a bound on the difference between the load of any machine $M_i$ at step~$j$, and the average load of all jobs in $J_j$.

\begin{corollary}
\label{cor:BoundedDeviation}%
With the same preconditions as in Theorem~\ref{thm:GenericStructure}, there exists an optimal proper schedule $\sigma$ with
$$
\Big|P(J_{i,j}(\sigma)) - \frac1m P(J_j)\Big| \;<\; \const p_{\max}^2
$$ 
for each $1 \leq i \leq m$ and~$1 \leq j  \leq n$.    
\end{corollary}

\begin{proof}
Using Theorem~\ref{thm:GenericStructure}, we have for each $1 \leq i \leq m$ and~$1 \leq j  \leq n$ that
\begin{align*}
\Big|P(J_{i,j}(\sigma)) - \frac1m P(J_j)\Big| &\;=\; 
\Big|P(J_{i,j}(\sigma)) - \frac1m \sum^m_{h=1}P(J_{h,j}(\sigma))\Big|\\
&\;=\;\Big|\frac1m \sum^m_{h=1} \Big(P(J_{i,j}(\sigma)) -  P(J_{h,j}(\sigma))\Big)\Big| \\
&\;\leq\; \frac{m-1}{m} \max_{1 \leq h \leq m} \Big| P(J_{i,j}(\sigma)) - P(J_{h,j}(\sigma)) \Big| \;<\; \const p^2_{\max}.
\end{align*}
The claim thus follows. \qed
\end{proof}

\section{The \boldmath{$Pm||\sum w_jC_j$} problem}

Below we apply our generic structure theorem to speed up the Lawler-Moore dynamic programming algorithm for the $Pm|| \sum w_jC_j$ problem. Recall that, by Smith's Rule~\cite{Smith1956}, the priority ordering in this problem is given by non-increasing efficiencies, and so we assume that $J$ is ordered such that $e_1 \geq e_2 \geq \cdots \geq e_n$. We first show that the swap described in Section~\ref{sec:Swap} is optimal; this will then allow us to use Theorem~\ref{thm:GenericStructure} to prune the state space of the Lawler-Moore dynamic program.

\subsection{Optimality of the swap}

Let $\sigma$ be any proper schedule for $J$. Recall that our swap produces an intermediate schedule~$\sigma'$ that is not necessarily proper, and then transforms $\sigma'$ into a proper schedule $\sigma''$ with no larger total weighted completion time. Thus, to argue that our swap is always optimal, it is enough to show that $\sum w_jC_j(\sigma') \leq \sum w_jC_j(\sigma)$. Below we prove this by first introducing a general lemma that will be useful later on. For brevity, we write $C_j=C_j(\sigma)$ and $C'_j=C_j(\sigma')$ for all $j \in J$.

\begin{lemma}
\label{lem:Swap-Cjs}
Let $\sigma'$ be the intermediate schedule obtained from~$\sigma$ by a swap at step~$j^* \in J$ between machines $M_h$ and $M_i$. Then the following hold:
\begin{enumerate}
\item $C'_j=C_j$ for all $j\notin J_H\cup J_I$, 
\item $C'_j \leq C_j$ for all $j\in J_H$, and
\item $C'_j \geq C_j$ for all $j\in J_{I}$.
\end{enumerate}
\end{lemma}

\begin{proof}
We argue each item of the lemma separately:
\begin{enumerate}
\item Let $j\notin J_H\cup J_I$. Only jobs in $J_H \cup J_I$ are moved by the swap, and so $C'_j=C_j$.
\item Let $j\in J_H$. If $j\in J_{H'}$, then in $\sigma$ it starts processing no earlier than time $t_i+P(J_I)$, while in $\sigma'$ it completes no later than that time. Hence~$C'_j<C_j$ in this case. If $j\in J_{H''}$, then the jobs $j' \in J_{H'} \neq \emptyset$ with $j' <j$ are now no longer processed before $j$ on machine $M_h$ in~$\sigma'$, so $C'_j \leq C_j$. 
\item Let $j \in J_{I}$. If $j \in J_{I'}$, then in~$\sigma$ it completes no later than time $t_i+P(J_I)$, while in~$\sigma'$ it is processed on machine~$M_h$ after all jobs of $J_{H''}$, and every job of $J_{H''}$ starts no earlier than time $t_i+P(J_I)$. Thus, job~$j$ starts in $\sigma'$ no earlier than time $t_i+P(J_I)$, and $C_j'>C_j$ in this case. If $j \in J_{I''}$, then in $\sigma$ it completes no later than time $t_i+P(J_{I'})+P(J^{\leq}_{I''})$, where $J^{\leq}_{I''}$ is the set of all jobs $j' \in J_{I''}$ with $j' \leq j$, whereas in $\sigma'$ it completes exactly at time  $t_i+P(J_{H'})+P(J^{\leq}_{I''})=t_i+P(J_{I'})+P(J^{\leq}_{I''})$. Hence, $C_j'\geq C_j$ in this case. \qed
\end{enumerate}
\end{proof}

\begin{lemma}
\label{lem:wjCj-SwapOpt}
$\sum w_j C'_j \;\leq\; \sum w_j C_j$.
\end{lemma}

\begin{proof}
Set $e:=e_{j^*}$, and define
$$
\delta_j\;:=\;w_j-ep_j\;=\;(e_j-e)p_j
$$
for all $j \in J$. For each $i \in \{1,\ldots,m\}$, let $J_i(\sigma)=\{j_{i,1},\ldots,j_{i,n(i)}\}$. Then as the final machine loads do not change from $\sigma$ to $\sigma'$, \emph{i.e.} $P(J_i(\sigma))=P(J_i(\sigma'))$ for all $i \in \{1,\ldots,m\}$, we have  
\begin{align*}
\sum_{j=1}^n p_j C_j \;&=\; \sum^m_{i=1}\sum_{j \in J_i(\sigma)} p_j C_j  \;=\; \sum^m_{i=1}\sum^{n(i)}_{k=1} p_{i,j_k} \sum^{k}_{\ell=1} p_{i,j_\ell}\\ 
\;&=\; \sum^m_{i=1}\frac{1}{2}\left(\left( \sum_{k=1}^{n(i)} p_{i,j_k} \right)^2 + \sum_{k=1}^{n(i)} p_{i,j_k}^2\right) \\
\;&=\; \sum^m_{i=1}\frac{1}{2}\left(P(J_i(\sigma))^2 + \sum_{k=1}^{n(i)} p_{i,j_k}^2\right) 
\;=\; \sum^m_{i=1}\frac{1}{2}\left(P(J_i(\sigma'))^2 + \sum_{k=1}^{n(i)} p_{i,j_k}^2\right) \;=\; \sum_{j=1}^n p_j C'_j,
\end{align*}
where the last equality follows from the same reasoning.
Therefore,
$$
\sum_{j=1}^n w_j(C'_j-C_j)
= e\sum_{j=1}^n p_j(C_j'-C_j)+\sum_{j=1}^n \delta_j(C_j'-C_j)
= \sum_{j=1}^n \delta_j(C_j'-C_j).
$$
Observe that by Lemma~\ref{lem:Swap-Cjs}, if $j\notin J_{H} \cup J_{I}$ then $C'_j=C_j$,  and so the contribution of~$j$ is zero to this last sum. If $j\in J_H$, then $j\le j^*$ and so $\delta_j\ge0$. By Lemma~\ref{lem:Swap-Cjs} we have $C'_j\leq C_j$, and therefore $\delta_j(C'_j -C_j) \leq 0$. If $j\in J_{I}$, then $j>j^*$ and hence $\delta_j\leq 0$. By Lemma~\ref{lem:Swap-Cjs} we have $C_j'\geq C_j$, and so $\delta_j(C'_j-C_j) \leq 0$. Consequently, every term in $\sum_j \delta_j(C_j'-C_j)$ is non-positive, and the lemma follows. \qed
\end{proof}

\subsection{Lawler-Moore speedup}

According to Corollary~\ref{cor:BoundedDeviation} of the theorem, we can restrict our attention to proper schedules $\sigma$ with 
$$
-\const p^2_{\max} \;\leq\; P(J_{i,j}(\sigma)) -  P(J_j)/m   \;\leq\; \const p^2_{\max}
$$ 
for all $1\leq i \leq m$ and $1 \leq j \leq n$. We can now apply the Lawler-Moore dynamic program, where the state space is restricted to $P_1,\ldots,P_{m-1} \in \{-4p^2_{\max},\ldots,4p^2_{\max}\}$. In our pruned version, a state $T_j[P_1,\ldots,P_{m-1}]$ is defined as the minimum total weighted completion time of any schedule for~$J_j$, where $P(J_j)/m+P_i$ represents the total processing time (load) assigned to machine $M_i$. This yields the following modified Lawler-Moore recursion:
$$
T_j[P_1,\ldots,P_{m-1}] = \min
\begin{cases}
\quad T_{j-1}[P'_1 -p_j,\ldots,P'_{m-1}] \,+\, w_j \cdot (P(J_j)/m+P_1),\\ 
\quad \vdots\\
\quad T_{j-1}[P'_1,\ldots,P'_{m-1} -p_j] +w_j \cdot (P(J_j)/m +P_{m-1}),\\
\quad T_{j-1}[P'_1,\ldots,P'_{m-1}] + w_j \cdot (P(J_j)/m- \sum^{m-1}_{i=1} P_{i}), \\
\end{cases}
$$
where $P'_i= P_i+p_j/m$ for all $1 \leq i \leq m$. Naturally, we set $T_j[P_1,\ldots,P_{m-1}]=\infty$ whenever some $P_i \notin \{-4p^2_{\max},\ldots,4p^2_{\max}\}$. Correctness of this recursion follows from the correctness of the original Lawler-Moore recurrence and Theorem~\ref{thm:GenericStructure}. The number of states in the program is $O(p_{\max}^{2m-2} \cdot n)$, and since computing each state requires $O(m)=O(1)$ time, the running time of the resulting dynamic program is $O(p_{\max}^{2m-2} \cdot n)$.

We next use the processing-time/weight duality of Hall and Chudak, as described by Goemans and Williamson~\cite[pp.~283--284]{GoemansWilliamson00}. Their observation is stated for a single machine; applying the same transformation independently on each machine gives the following lemma.

\begin{lemma}
\label{lem:WCT-duality}%
For an instance $\mathcal I$ of $Pm||\sum w_jC_j$, let $\mathcal I^\dagger$ be obtained by setting $p_j^\dagger=w_j$ and $w_j^\dagger=p_j$ for every job~$j$. Reversing the job sequence on each machine defines a value-preserving bijection between non-idling schedules for $\mathcal I$ and $\mathcal I^\dagger$. In particular, $\operatorname{OPT}(\mathcal I)=\operatorname{OPT}(\mathcal I^\dagger)$, and the bijection maps Smith-ordered schedules to Smith-ordered schedules.
\end{lemma}
\begin{proof}
Consider a machine whose job sequence in a schedule for $\mathcal I$ is $(j_1,\ldots,j_q)$. Its contribution to the objective is
\begin{align*}
\sum_{t=1}^q w_{j_t}C_{j_t}
&=\sum_{t=1}^q w_{j_t}\sum_{s=1}^t p_{j_s}
=\sum_{s=1}^q p_{j_s}\sum_{t=s}^q w_{j_t}.
\end{align*}
In the reversed sequence $(j_q,\ldots,j_1)$ for $\mathcal I^\dagger$, job~$j_s$ has completion time $C_{j_s}^\dagger=\sum_{t=s}^q p_{j_t}^\dagger=\sum_{t=s}^q w_{j_t}$. Hence, the last expression above is $\sum_{s=1}^q w_{j_s}^\dagger C_{j_s}^\dagger$. Summing over the machines proves that the transformation preserves the objective value. Since applying it twice returns the original instance and schedule, it is a bijection. Moreover, as neither instance has release dates, idle time can be removed without increasing the objective, so both instances admit an optimal non-idling schedule and have the same optimum value. Finally,
\[
\frac{w_j^\dagger}{p_j^\dagger}=\frac{p_j}{w_j}=\frac{1}{e_j}.
\]
Thus, reversing a non-increasing order of the original efficiencies yields a non-increasing order of the dual efficiencies, so the transformation preserves Smith ordering.
\end{proof}

The maximum processing time of $\mathcal I^\dagger$ is $w_{\max}$. Since the input is preordered by Smith's Rule, its reverse is a valid priority ordering for $\mathcal I^\dagger$ and can be produced in $O(n)$ time. We may therefore run the dynamic program above on whichever of $\mathcal I$ and $\mathcal I^\dagger$ has smaller maximum processing time and transform the resulting schedule back using Lemma~\ref{lem:WCT-duality}.

\begin{theorem}
$Pm||\sum w_jC_j$ can be solved in $O(\min\{p_{\max},w_{\max}\}^{2m-2} \cdot n)$ time.
\end{theorem}

\section{The \boldmath{$Pm||L_{\max}$} problem}

In the $Pm||L_{\max}$ problem, the priority order is in non-decreasing due dates $d_1 \leq d_2 \leq \cdots \leq d_n$, due to Jackson's Rule~\cite{Jackson1955}. Below we show that our framework applies to this problem as well. As in the previous section, we begin by proving that our swap is optimal for $L_{\max}$, and then present the pruned Lawler-Moore recursion for $Pm||L_{\max}$. 

\begin{lemma}
\label{lem:Swap-Cjs2}
Let $\sigma'$ be the intermediate schedule obtained from~$\sigma$ by a swap at step~$j^* \in J$ between machines $M_h$ and $M_i$. Then $C'_j\leq C_{j^*}$ for all $j\in J_{H} \cup J_{I}$.
\end{lemma}

\begin{proof}
By construction, all swapped jobs finish at the latest at time $t_i+\Delta_{h,i,j^*}= t_h=C_{j^*}$. \qed 
\end{proof}

\begin{lemma}
\label{lem:Lmax-SwapOpt}
$L_{\max}(\sigma') \leq L_{\max}(\sigma)$.
\end{lemma}

\begin{proof}
Suppose $\sigma'$ is obtained from $\sigma$ by a swap at index $j^* \in J$. For each job $j \in J$, write $L_j:=C_j-d_j$ and $L'_j:=C'_j-d_j$. If $j\in J_H$, then by Lemma~\ref{lem:Swap-Cjs} we have $C'_j \leq C_j$, and therefore $L'_j \leq L_j\leq L_{\max}(\sigma)$. If $j\in J_I$, then $j>j^*$ and so $d_j\geq d_{j^*}$. By Lemma~\ref{lem:Swap-Cjs2},
$$
L'_j \;=\; C'_j-d_j \;\leq\; C_{j^*}-d_{j^*} \;=\; L_{j^*} \leq L_{\max}(\sigma).
$$
If $j\notin J_H\cup J_I$, then Lemma~\ref{lem:Swap-Cjs} gives $L'_j=L_j\le L_{\max}(\sigma)$. Thus $L'_j \leq L_{\max}(\sigma)$ for every job~$j\in J$, and the lemma follows. \qed
\end{proof}

The Lawler-Moore dynamic program for $Pm||L_{\max}$ formulates a state $T_j[P_1,\ldots,P_{m-1}]$ as the minimum maximum lateness $L_{\max}(\sigma)$ among all proper schedules~$\sigma$ for~$J_j$ where the total load on machine~$M_i$ is~$P_i$, for each $1 \leq i \leq m-1$. Due to Lemma~\ref{lem:Lmax-SwapOpt}, Theorem~\ref{thm:GenericStructure} and Corollary~\ref{cor:BoundedDeviation} apply to $L_{\max}$. Thus, we can again restrict the state variables of the dynamic program to values $P_1,\ldots,P_{m-1} \in \{-4p^2_{\max},\ldots,4p^2_{\max}\}$, where $T_j[P_1,\ldots,P_{m-1}]$ now corresponds to schedules where the total load on machine~$M_i$ is~$P(J_j)/m+P_i$. As in the previous section, we may assume that the processing time of each job is a multiple of~$m$. The resulting pruned Lawler-Moore recursion is given by:
$$
T_j[P_1,\ldots,P_{m-1}] = \min
\begin{cases}
\quad \max \Big\{T_{j-1}[P'_1 - p_j,\ldots,P'_{m-1}],\, P(J_j)/m+P_1 - d_j\Big\},\\ 
\quad \vdots\\
\quad \max \Big\{T_{j-1}[P'_1,\ldots,P'_{m-1} - p_j],\, P(J_j)/m + P_{m-1} - d_j\Big\},\\
\quad \max \Big\{T_{j-1}[P_1,\ldots,P_{m-1}],\, P(J_j)/m  -  \sum^{m-1}_{i=1} P_{i} - d_j \Big\}, \\
\end{cases}
$$
where $P'_i = P_i+p_j/m$ for all $1 \leq i \leq m-1$. Correctness of this recursion follows from the correctness of the original Lawler-Moore recurrence and Theorem~\ref{thm:GenericStructure}. The running time of the resulting dynamic program is determined by the number of states, which is $O(p^{2m-2}_{\max} \cdot n)$. 
\begin{theorem}
$Pm|| L_{\max}$ can be solved in $O(p^{2m-2}_{\max} \cdot n)$ time.    
\end{theorem}

\section{The \boldmath{$Pm||\sum w_jU_j$} problem}

As a final application of our framework, we next consider the $Pm||\sum w_jU_j$ problem. Here the jobs are ordered according to Jackson's Rule, and so we assume that $d_1 \leq d_2 \leq \cdots \leq d_n$. For the $\sum w_jU_j$ objective, it
is convenient to use the equivalent formulation in which a tardy job is simply discarded and incurs a penalty of~$w_j$. Thus a feasible solution consists of a subset
of scheduled jobs $J'\subseteq J$, processed according to Jackson's Rule, and every scheduled job must complete by its due date.

\begin{lemma}
\label{lem:wjUj-SwapOpt}
Let $\sigma'$ be the intermediate schedule obtained from~$\sigma$. Then $\sum w_jU_j(\sigma') \leq \sum w_jU_j(\sigma)$.
\end{lemma}

\begin{proof}
Let $\sigma'$ be obtained from $\sigma$ by a swap at index $j^*$. The swap does not change the set of discarded jobs, so it suffices to show that every job scheduled in $\sigma'$ still completes by its due date. If $j\in J_H$, then Lemma~\ref{lem:Swap-Cjs} yields
$C'_j\le C_j\le d_j$, so $j$ remains early. If $j\in J_I$, then $j>j^*$ and $d_j\ge d_j^*$. Since job~$j$ is scheduled in $\sigma$, we have $C_j \le d_j$. Therefore, by Lemma~\ref{lem:Swap-Cjs2}, we have
$$
C'_j \;\leq\; C_{j^*} \;\leq\; d_{j^*} \;\leq\; d_j,
$$
and so~$j$ is also on time in $\sigma'$. If $j\notin J_H\cup J_I$, then Lemma~\ref{lem:Swap-Cjs} gives $C'_j=C_j\le d_j$. Thus every scheduled job in $\sigma'$ is early, and the set of discarded jobs is unchanged. The lemma thus follows. \qed
\end{proof}

Due to Lemma~\ref{lem:wjUj-SwapOpt}, our generic structure theorem (Theorem~\ref{thm:GenericStructure}) applies to the $\sum w_jU_j$ objective as well. Here we diverge from the previous algorithms by using the variables $P_i' := P_i-P_m$ for $1 \leq i \leq m-1$, and also keeping track of the absolute load~$P_m$, which is necessary since discarded jobs make the total scheduled load variable. By Theorem~\ref{thm:GenericStructure}, we may restrict our attention to proper schedules~$\sigma$ with $\lvert P_i' \rvert < \const p^2_{\max}$ for all $1 \leq i \leq m-1$ and $1 \leq j \leq n$. Accordingly, we define a state $T_j[P_1',\ldots,P_{m-1}',P_m]$ as the minimum total weighted number of tardy jobs among all proper schedules~$\sigma$ for~$J_j$ such that $P(J_{i,j}(\sigma)) - P(J_{m,j}(\sigma)) = P_i'$ for each $1 \leq i \leq m-1$, and $P(J_{m,j}(\sigma)) = P_m$. The resulting pruned Lawler-Moore recursion is
$$
T_j[P_1',\ldots,P_{m-1}',P_m] = \min
\begin{cases}
\quad \begin{cases}
T_{j-1}[P_1' - p_j,\ldots,P_{m-1}',P_m] \;&:\; P_m+P_1' \leq d_j,\\
\infty \;&:\; \text{otherwise}.
\end{cases}\\
\quad \vdots\\
\quad \begin{cases}
T_{j-1}[P_1',\ldots,P_{m-1}'-p_j,P_m] \;&:\; P_m+P_{m-1}' \leq d_j,\\
\infty \;&:\; \text{otherwise}.
\end{cases}\\
\quad \begin{cases}
T_{j-1}[P_1'+p_j,\ldots,P_{m-1}'+p_j,P_m-p_j] \;&:\; P_m \leq d_j,\\
\infty \;&:\; \text{otherwise}.
\end{cases}\\
\quad T_{j-1}[P_1',\ldots,P_{m-1}',P_m] + w_j, \\
\end{cases}
$$
Here the first $m-1$ cases correspond to scheduling job~$j$ early on one of the machines $M_1,\ldots,M_{m-1}$, the next case corresponds to scheduling it early on~$M_m$, and the last case corresponds to discarding~$j$ and paying penalty~$w_j$. Naturally, we set $T_j[P_1',\ldots,P_{m-1}',P_m]=\infty$ whenever some $P_i' \notin \{-\const p^2_{\max},\ldots,\const p^2_{\max}\}$ or $P_m \notin \{0,\ldots,P\}$. Correctness of the recursion follows from the correctness of the original Lawler-Moore recurrence and Theorem~\ref{thm:GenericStructure}. The number of states, and hence the running time of the resulting dynamic program, is $O(p^{2m-2}_{\max} \cdot P \cdot n)$. 
\begin{theorem}
$Pm||\sum w_jU_j$ can be solved in $O(p^{2m-2}_{\max} \cdot P \cdot n)$ time.   
\end{theorem}

\section{Discussion}


In this paper, we generalize the fine-grained parameterization introduced by Klein, Polak, and Rohwedder~\cite{Klein0R23} by replacing the traditional pseudo-polynomial dependence on the total processing time P with a dependence on the maximum job size~$p_{\max}$ across the broader Lawler-Moore framework. For total weighted completion time, the processing-time/weight duality further strengthens this dependence to $\min\{p_{\max},w_{\max}\}$. Our results provide significant algorithmic speedups for fundamental parallel-machine problems and show that the SETH-based lower bounds previously associated with~$P$ can be bypassed beyond the specific case of minimizing tardy processing times. Moreover, we obtain a generic structural result on optimal solutions that extends the applicability of prefix-load bounding to complex sequencing and arbitrary weight functions. We believe that our generalized swapping framework is robust enough to accelerate a wide range of other dynamic programming recurrences in scheduling and combinatorial optimization.

Several interesting directions for future research remain. Below we provide a partial list of those we believe are the most important (in no particular order):
\begin{enumerate}
\item Can one solve $Pm||\sum w_jU_j$ in $w_{\max}^{O(1)}\cdot n$ time? The processing-time/weight duality for total weighted completion time does not extend to the weighted number of tardy jobs, so a positive answer appears to require a different approach.
\item Can $P2||\sum w_jC_j$ be solved in $\widetilde{O}(\min\{p_{\max},w_{\max}\}^2+n)$ time? For $P2||C_{\max}$, a running time of $\widetilde{O}(p_{\max}^2+n)$ follows from recent Knapsack algorithms~\cite{Bringmann24,Jin24}.
\item Assuming $\forall\exists$-SETH, Abboud et al.~\cite{AbboudBHS22} rule out $\widetilde{O}(n+p_{\max}\cdot n^{1-\varepsilon})$-time algorithms, for every fixed $\varepsilon>0$, for several scheduling problems on one or two machines, including $P2||L_{\max}$ and $1||\sum w_jU_j$. Their reductions also rule out $\widetilde{O}(n+p_{\max}^c)$-time algorithms for these problems for every fixed constant $c>0$. Can these lower bounds be strengthened, or extended to problems with more than two machines?
\item What about the $Rm$ machine model, where the~$p_{i,j}$'s are arbitrary. Can one get algorithms with running times of the form $p_{\max}^{O(1)}\cdot n$ in this environment as well? 
\end{enumerate}

\bibliographystyle{alpha}
\bibliography{biblo}

\newcommand{\etalchar}[1]{$^{#1}$}
\begin{thebibliography}{BFH{\etalchar{+}}22}

\bibitem[ABHS22]{AbboudBHS22}
Amir Abboud, Karl Bringmann, Danny Hermelin, and Dvir Shabtay.
\newblock Scheduling lower bounds via {AND} subset sum.
\newblock {\em Journal of Computer and System Sciences}, 127:29--40, 2022.

\bibitem[BDW26]{BringmannDW26}
Karl Bringmann, Anita D{\"u}rr, and Karol W\k{e}grzycki.
\newblock Tight {(S)ETH}-based lower bounds for pseudopolynomial algorithms for bin packing and multi-machine scheduling.
\newblock {\em CoRR}, abs/2603.12999, 2026.

\bibitem[BFH{\etalchar{+}}22]{BringmannFHSW22}
Karl Bringmann, Nick Fischer, Danny Hermelin, Dvir Shabtay, and Philip Wellnitz.
\newblock Faster minimization of tardy processing time on a single machine.
\newblock {\em Algorithmica}, 84(5):1341--1356, 2022.

\bibitem[BJS74]{BrunoCoffman}
John Bruno, Edward G.~Coffman Jr., and Ravi Sethi.
\newblock Scheduling independent tasks to reduce mean finishing time.
\newblock {\em Communications of the {ACM}}, 17:382--387, 1974.

\bibitem[Bri24]{Bringmann24}
Karl Bringmann.
\newblock Knapsack with small items in near-quadratic time.
\newblock In {\em Proc. of the 56th Annual {ACM} Symposium on Theory Of Computing ({STOC})}, pages 259--270, 2024.

\bibitem[CMM67]{ConwayMM1967}
Richard~W. Conway, William~L. Maxwell, and Leslie~W. Miller.
\newblock {\em Theory of Scheduling}.
\newblock Addison-Wesley, 1967.

\bibitem[FW25]{FischerWennmann25}
Nick Fischer and Leo Wennmann.
\newblock Minimizing tardy processing time on a single machine in near-linear time.
\newblock {\em TheoretiCS}, 4, 2025.

\bibitem[GJ78]{GJ78}
Michael~R. Garey and David~S. Johnson.
\newblock {S}trong {N}{P}-completeness results: motivation, examples, and implications.
\newblock {\em Journal of the ACM}, 25(3):499--508, 1978.

\bibitem[GJ90]{GareyJohnson}
Michael~R. Garey and David~S. Johnson.
\newblock {\em Computers and Intractability; A Guide to the Theory of NP-Completeness}.
\newblock W. H. Freeman \& Co., 1990.

\bibitem[Gra69]{Graham1969}
Ronald~L. Graham.
\newblock Bounds on multiprocessing timing anomalies.
\newblock {\em SIAM Journal on Applied Mathematics}, 17(2):416--429, 1969.

\bibitem[GW00]{GoemansWilliamson00}
Michel~X. Goemans and David~P. Williamson.
\newblock Two-dimensional gantt charts and a scheduling algorithm of {L}awler.
\newblock {\em {SIAM} Journal on Discrete Mathematics}, 13(3):281--294, 2000.

\bibitem[HH24]{HeegerH24}
Klaus Heeger and Danny Hermelin.
\newblock Minimizing the weighted number of tardy jobs is {W}[1]-hard.
\newblock In {\em Proc. of the 32nd Annual European Symposium on Algorithms ({ESA})}, pages 68:1--68:14, 2024.

\bibitem[HKPS21]{HermelinKPShabtay21}
Danny Hermelin, Shlomo Karhi, Michael~L. Pinedo, and Dvir Shabtay.
\newblock New algorithms for minimizing the weighted number of tardy jobs on a single machine.
\newblock {\em Annals of Operations Research}, 298(1):271--287, 2021.

\bibitem[HMS24]{HermelinMS24}
Danny Hermelin, Hendrik Molter, and Dvir Shabtay.
\newblock Minimizing the weighted number of tardy jobs via (max,+)-convolutions.
\newblock {\em {INFORMS} Journal of Computing}, 36(3):836--848, 2024.

\bibitem[HS26]{HermelinShabtay2026}
Danny Hermelin and Dvir Shabtay.
\newblock Fast makespan minimization via short {ILP}s.
\newblock {\em CoRR}, abs/2502.13631, 2026.

\bibitem[Jac55]{Jackson1955}
James~R. Jackson.
\newblock Scheduling a production line to minimize maximum tardiness.
\newblock {\em Management Science Research Project, University of California, Los Angeles, CA.}, 1955.

\bibitem[Jin24]{Jin24}
Ce~Jin.
\newblock 0-1 knapsack in nearly quadratic time.
\newblock In {\em Proc. of the 56th Annual {ACM} Symposium on Theory Of Computing (STOC)}, pages 271--282, 2024.

\bibitem[KK18]{KnopKoutecky18}
Dusan Knop and Martin Kouteck{\'{y}}.
\newblock Scheduling meets $n$-fold integer programming.
\newblock {\em Journal of Scheduling}, 21(5):493--503, 2018.

\bibitem[KPR23]{Klein0R23}
Kim{-}Manuel Klein, Adam Polak, and Lars Rohwedder.
\newblock On minimizing tardy processing time, max-min skewed convolution, and triangular structured {ILP}s.
\newblock In {\em Proc. of the 34th {ACM-SIAM} Symposium on Discrete Algorithms (SODA)}, pages 2947--2960, 2023.

\bibitem[LLRK76]{LagewegEtAl1976}
B.~J. Lageweg, J.~K. Lenstra, and A.~H.~G. Rinnooy~Kan.
\newblock Minimizing maximum lateness on one machine: Computational experience and some applications.
\newblock {\em Management Science}, 23(10):1151--1163, 1976.

\bibitem[LM69]{LawlerMoore}
Eugene~L. Lawler and James~M. Moore.
\newblock A functional equation and its application to resource allocation and sequencing problems.
\newblock {\em Management Science}, 16(1):77--84, 1969.

\bibitem[LRKB77]{LenstraEtAl77}
J.K. Lenstra, A.H.G. Rinnooy~Kan, and P.~Brucker.
\newblock Complexity of machine scheduling problems.
\newblock {\em Annals of Discrete Mathematics}, 1:343--362, 1977.

\bibitem[Moo68]{Moore1968}
James~M. Moore.
\newblock An $n$ job, one machine sequencing algorithm for minimizing the number of late jobs.
\newblock {\em Management Science}, 15:102--109, 1968.

\bibitem[PRW21]{PolakEtAl21}
Adam Polak, Lars Rohwedder, and Karol W\k{e}grzycki.
\newblock Knapsack and subset sum with small items.
\newblock In {\em Proc. of the 48th International Colloquium on Automata, Languages, and Programming (ICALP)}, pages 106:1--106:19, 2021.

\bibitem[Smi56]{Smith1956}
Wayne~E. Smith.
\newblock Various optimizers for single-stage production.
\newblock {\em Naval Research Logistics Quarterly}, 3:59--66, 1956.

\bibitem[SS23]{SchieberS23}
Baruch Schieber and Pranav Sitaraman.
\newblock Quick minimization of tardy processing time on a single machine.
\newblock In {\em Proc. of the 18th International Symposium on Algorithms and Data Structures ({WADS})}, pages 637--643, 2023.

\end{thebibliography}

\end{document}